\documentclass{llncs}
\usepackage{graphicx}
\usepackage{hyperref}
\usepackage{xcolor}
\usepackage{xspace}
\usepackage{paralist}
\usepackage{amsmath,amssymb}
\usepackage{subfig}
\usepackage{wrapfig}

\newtheorem{observation}{Observation}

\newcommand{\NP}{$\mathcal{NP}$\xspace}

\renewcommand{\NP}{$\mathcal{NP}$}

\newcommand{\NPC}{\mbox{NP-complete}\xspace}

\definecolor{blue}{rgb}{0.274,0.392,0.666}
\definecolor{red}{rgb}{0.627,0.117,0.156}
\definecolor{green}{rgb}{0,0.588,0.509}

\newcommand{\red}[1]{{\color{red}{#1\xspace}}}
\newcommand{\blue}[1]{{\color{blue}{#1\xspace}}}
\newcommand{\green}[1]{{\color{green}{#1\xspace}}}

\newcommand{\Gr}[1]{$\red{G^{#1}_1}$\xspace}
\newcommand{\Gb}[1]{$\blue{G^{#1}_2}$\xspace}
\newcommand{\Gg}[1]{$\green{G^{#1}_3}$\xspace}

\newcommand{\sefeinstance}[1]{$\langle\red{G^{#1}_1},\blue{G^{#1}_2}\rangle$\xspace}
\newcommand{\sefethreeinstance}[1]{$\langle\red{G^{#1}_1},\blue{G^{#1}_2},\green{G^{#1}_3}\rangle$\xspace}

\let\doendproof\endproof
\renewcommand\endproof{~\hfill\qed\doendproof}

\newcommand{\try}{Beyond Level Planarity}
\title{\try\thanks{Research was partially supported by DFG grant Ka812/17-1, by MIUR project AMANDA, prot. 2012C4E3KT\_001, and by DFG grant WA 654/21-1.}
}

\newcommand{\tuba}{$^{\dag}$} 
\newcommand{\rome}{$^{\diamond}$} 
\newcommand{\kit}{$^{\circ}$}

\author{Patrizio {Angelini\tuba}, Giordano {Da Lozzo\rome}, Giuseppe {Di Battista\rome}, \\
    Fabrizio {Frati\rome}, Maurizio {Patrignani\rome}, Ignaz {Rutter\kit}
	\institute{
    \tuba T\"ubingen University, Germany --
    \email{angelini@informatik.uni-tuebingen.de}\\
	\rome~Roma Tre University, Italy --
    \email{\{dalozzo,gdb,frati,patrigna\}@dia.uniroma3.it}\\
    \kit~Karlsruhe Institute of Technology, Germany -- 
    \email{rutter@kit.edu}
}}

\begin{document}
\maketitle
\pagestyle{plain}

\begin{abstract}
In this paper we settle the computational complexity of two open problems related to the extension of the notion of level planarity to surfaces different from the plane. Namely, we show that the problems of testing the existence of a level embedding of a level graph on the surface of the rolling cylinder or on the surface of the torus, respectively known by the name of {\sc Cyclic Level Planarity} and {\sc Torus Level Planarity}, are polynomial-time solvable. 
 
Moreover, we show a complexity dichotomy for testing the {\sc Simultaneous Level Planarity} of a set of level graphs, with respect to both the number of level graphs and the number of levels.
\end{abstract}


\section{Introduction and Overview} \label{se:intro}

The study of level drawings of level graphs has spanned a long time; the seminal paper by Sugiyama {\em et al.}~on this subject~\cite{stt-mvuhs-81} dates back to 1981, well before graph drawing was recognized as a distinguished research area. This is motivated by the fact that level graphs naturally model hierarchically organized data sets and level drawings are a very intuitive way to represent such graphs.

Formally, a \emph{level graph} $(V,E,\gamma)$ is a directed graph $(V,E)$ together with a function $\gamma: V \rightarrow \{1,...,k\}$, with $1 \leq k \leq |V|$. The set $V_i = \{ v \in V \colon \gamma(v)=i\}$ is the $i$-th \emph{level} of $(V,E,\gamma)$.  
A level graph $(V,E,\gamma)$ is \emph{proper} if for each $(u,v) \in E$, either $\gamma(u) = \gamma(v) - 1$, or $\gamma(u)=k$ and $\gamma(v)=1$. Let $l_1,\dots,l_k$ be $k$ horizontal straight lines on the plane ordered in this way with respect to the $y$-axis. A \emph{level drawing} of $(V,E,\gamma)$ maps each vertex $v\in V_i$ to a point on $l_i$ and each edge $(u,v)\in E$ to a curve monotonically increasing in the $y$-direction from $u$ to $v$. Note that a level graph $(V,E,\gamma)$ containing an edge $(u,v)\in E$ with $\gamma(u)>\gamma(v)$ does not admit any level drawing. A level graph is \emph{level planar} if it admits a {\em level embedding}, i.e., a level drawing with no crossing; see Fig.~\ref{fig:drawings}(a). The {\sc Level Planarity} problem asks to test whether a given level graph is level planar.
 
\begin{figure}[tb!]
    \centering
      \subfloat[]{\includegraphics[height=.25\textwidth,page=1]{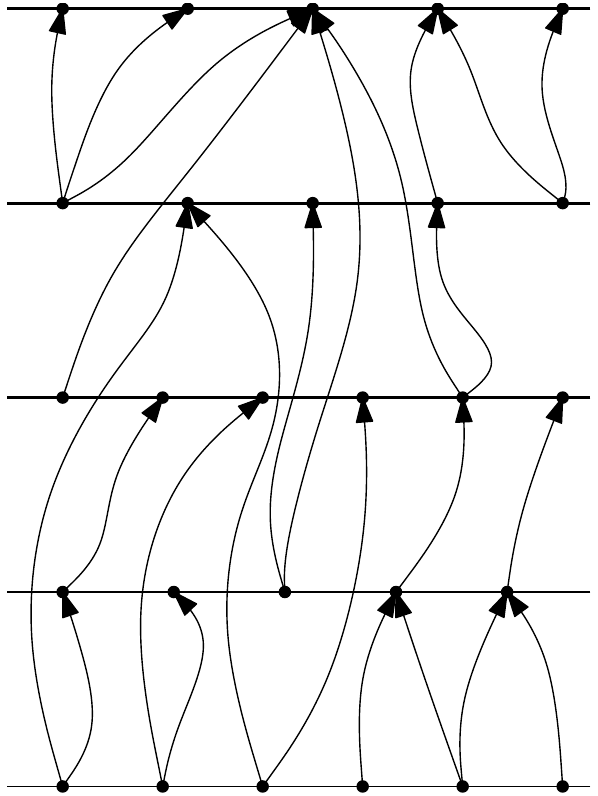}}\hfill
      \subfloat[]{\includegraphics[height=.25\textwidth,page=2]{level.pdf}}\hfill
      \subfloat[]{\includegraphics[height=.25\textwidth,page=3]{level.pdf}}\hfill
      \subfloat[]{\includegraphics[height=.25\textwidth,page=4]{level.pdf}}
    \caption{Level embeddings (a) on the plane, (b) on the standing cylinder, (c) on the rolling cylinder, and (d) on the torus.}
\label{fig:drawings}
\end{figure}

Problem {\sc Level Planarity} has been studied for decades~\cite{dn-hpt-88,hp-rlpdlt-95,jlm-lptlt-98,rsbhkmsc-sfplg-01,Fulek2013}, starting from a characterization of the single-source level planar graphs~\cite{dn-hpt-88} and culminating in a linear-time algorithm for general level graphs~\cite{jlm-lptlt-98}.
A characterization of level planarity in terms of ``minimal'' forbidden subgraphs is still missing~\cite{efk-clptmp-09,hkl-clpg-04}.
The problem has also been studied to take into account a clustering of the vertices ({\sc Clustered Level Planarity}~\cite{tibp-addfr-15,fb-clp-04}) or consecutivity constraints for the vertex orderings on the levels ({\sc $\mathcal{T}$-level Planarity}~\cite{tibp-addfr-15,wsp-gktlg-12}).

Differently from the standard notion of planarity, the concept of level planarity does not immediately extend to representations of level graphs on surfaces different from the plane\footnote{We consider connected orientable surfaces; the {\em genus} of a surface is the maximum number of cuttings along non-intersecting closed curves that do not disconnect it.}. When considering the surface $\mathbb O$ of a sphere, level drawings are usually defined as follows: The vertices have to be placed on the $k$ circles given by the intersection of $\mathbb O$ with $k$ parallel planes, and each edge is a curve on $\mathbb O$ that is monotone in the direction orthogonal to these planes. The notion of level planarity in this setting goes by the name of {\sc Radial Level Planarity} and is known to be decidable in linear time~\cite{bbf-rlptelt-05}. This setting is equivalent to the one in which the level graph is embedded on the ``standing cylinder'': Here, the vertices have to be placed on the circles defined by the intersection of the cylinder surface $\mathbb S$ with planes parallel to the cylinder bases, and the edges are curves on $\mathbb S$ monotone with respect to the cylinder axis; see~\cite{abbg-cpue-12,bbf-rlptelt-05,b-updsrc-14} and Fig.~\ref{fig:drawings}(b).

Problem {\sc Level Planarity} has been also studied on the surface $\mathbb R$ of a ``rolling cylinder''; see~\cite{abbg-cpue-12,bb-ltpte-08,bbk-clpte-07,b-updsrc-14} and Fig.~\ref{fig:drawings}(c). 
In this setting, $k$ straight lines $l_1,\dots,l_k$ parallel to the cylinder axis lie on $\mathbb R$, where $l_1,\dots,l_k$ are seen in this clockwise order from a point $p$ on one of the cylinder bases, the vertices of level $V_i$ have to be placed on $l_i$, for $i=1,\dots,k$, and each edge $(u,v)$ is a curve $\lambda$ lying on $\mathbb R$ and flowing monotonically in clockwise direction from $u$ to $v$ as seen from $p$.
%
Within this setting, the problem takes the name of {\sc Cyclic Level Planarity}~\cite{bbk-clpte-07}.
Note that a level graph $(V,E,\gamma)$ may now admit a level embedding even if it contains edges $(u,v)$ with $\gamma(u)>\gamma(v)$.
Contrary to the other mentioned settings, the complexity of testing {\sc Cyclic Level Planarity} is still unknown, and a polynomial (in fact, linear) time algorithm has been presented only for {\em strongly connected graphs}~\cite{bb-ltpte-08}, which are level graphs such that for each pair of vertices there exists a directed cycle through them. 

In this paper we settle the computational complexity of {\sc Cyclic Level Planarity} by showing a polynomial-time algorithm to test whether a level graph admits a cyclic level embedding (Theorem~\ref{co:cyclicLevel-polynomial}). 
In order to obtain this result, we study a version of level planarity in which the surface~$\mathbb T$ where the level graphs have to be embedded has genus $1$; we call {\sc Torus Level Planarity} the corresponding decision problem, whose study was suggested in~\cite{bbk-clpte-07}. It is not difficult to note (Lemmata~\ref{le:radial-torus} and~\ref{le:cyclic-torus}) that the torus surface combines the representational power of the surfaces of the standing and of the rolling cylinder -- that is, if a graph admits a level embedding on one of the latter surfaces, then it also admits a level embedding on the torus surface. Furthermore, both {\sc Radial Level Planarity} and {\sc Cyclic Level Planarity} (and hence {\sc Level Planarity}) reduce in linear time to {\sc Torus Level Planarity}.

The main result of the paper (Theorem~\ref{th:torus}) is a quadratic-time algorithm for proper instances of {\sc Torus Level Planarity} and a quartic-time algorithm for general instances. Our solution is based on a linear-time reduction (Observation~\ref{prop:characterization} and Lemmata~\ref{le:torus-characterization}-\ref{le:reduction}) that, starting from any proper instance of {\sc Torus Level Planarity}, produces an equivalent instance of the {\sc Simultaneous PQ-Ordering} problem~\cite{br-spacep-13} that can be solved in quadratic time (Theorem~\ref{th:2-fixed}).

Motivated by the growing interest in simultaneous embeddings of multiple planar graphs, which allow to display several relationships on the same set of entities in a unified representation, we define a new notion of level planarity in which multiple level graphs are considered and the goal is to find a simultaneous level embedding of them. 
The problem {\sc Simultaneous Embedding} (see the seminal paper~\cite{bcdeeiklm-spge-07} and a recent survey~\cite{bkr-sepg-13}) takes as input $k$ planar graphs $(V,E_1),\dots,(V,E_k)$ and asks whether they admit planar drawings mapping each vertex to the same point of the plane. We introduce the problem {\sc Simultaneous Level Planarity}, which asks whether $k$ level graphs $(V,E_1,\gamma),\dots,(V,E_k,\gamma)$ admit level embeddings mapping each vertex to the same point along the corresponding level.
As an instance of {\sc Simultaneous Level Planarity} for two graphs on two levels is equivalent to one of {\sc Cyclic Level Planarity} on two levels (Theorem~\ref{th:sim-poly}), we can solve {\sc Simultaneous Level Planarity} in polynomial time in this case. This positive result cannot be extended (unless P=NP), as the problem becomes \NPC even for two graphs on three levels and for three graphs on two levels (Theorem~\ref{thm:sim-level-npc-two-levels}). Altogether, this establishes a tight border of tractability for {\sc Simultaneous Level Planarity}.

\section{Preliminaries} \label{se:preliminaries}

A {\em tree} $T$ is a connected acyclic graph. The degree-$1$ vertices of $T$ are the {\em leaves} of $T$, denoted by ${\cal L}(T)$, while the remaining vertices are the {\em internal} vertices. 

A digraph $G=(V,E)$ without directed cycles is a {\em directed acyclic graph} ({\em DAG}). 
An edge $(u, v) \in E$ directed from $u$ to $v$ is an \emph{arc}; vertex $u$ is a {\em parent} of $v$ and $v$ is a {\em child} of $u$.
A vertex is a {\em source} ({\em sink}) if it has no parents~(children).

\smallskip
\noindent
{\bf Embeddings on levels.} 
An {\em embedding} of a graph on a surface $\mathbb Q$ is a mapping $\Gamma$ of each vertex $v$ to a distinct point on $\mathbb Q$ and of each edge $e=(u,v)$ to a simple Jordan curve on $\mathbb Q$ connecting $u$ and $v$, such that no two curves cross except at a common endpoint. 
Let $\mathbb I$ and $\mathbb S^1$ denote the unit interval and the boundary of the unit disk, respectively.
We define the surface $\mathbb S$ of the standing cylinder, $\mathbb R$ of the rolling cylinder, and $\mathbb T$ of the torus
as $\mathbb S = \mathbb S^1 \times \mathbb I$,
as $\mathbb R = \mathbb I \times \mathbb S^1$, and 
as $\mathbb T = \mathbb S^1 \times \mathbb S^1$, respectively.
The {\em $j$-th level} of surfaces $\mathbb S$, $\mathbb R$, and $\mathbb T$ with $k$ levels is defined as
 $l_j = \mathbb S^1 \times \{\frac{j-1}{k-1}\}$,
  $l_j = \mathbb I \times \{e^{2\pi i \frac{j-1}{k}}\}$, and
 $l_j = \mathbb S^1 \times \{e^{2\pi i \frac{j-1}{k}}\}$, respectively; see Fig.~\ref{fig:levels}.
%
%
An edge $(x,y)$ on $\mathbb S$, on $\mathbb R$, or on $\mathbb T$ is {\em monotone} if it intersects the levels $\gamma(x), \gamma(x)+1, \dots, \gamma(y)$, where $k+1=1$, exactly once and does not intersect any of the other levels.
%

 \begin{figure}[tb!]
\centering
\subfloat[]{\includegraphics[height=0.31\textwidth,page=3]{levels.pdf}\label{fig:levels-S}}\hfil
\subfloat[]{\includegraphics[height=0.31\textwidth,page=2]{levels.pdf}\label{fig:levels-R}}\hfil
\subfloat[]{\includegraphics[height=0.31\textwidth,page=1]{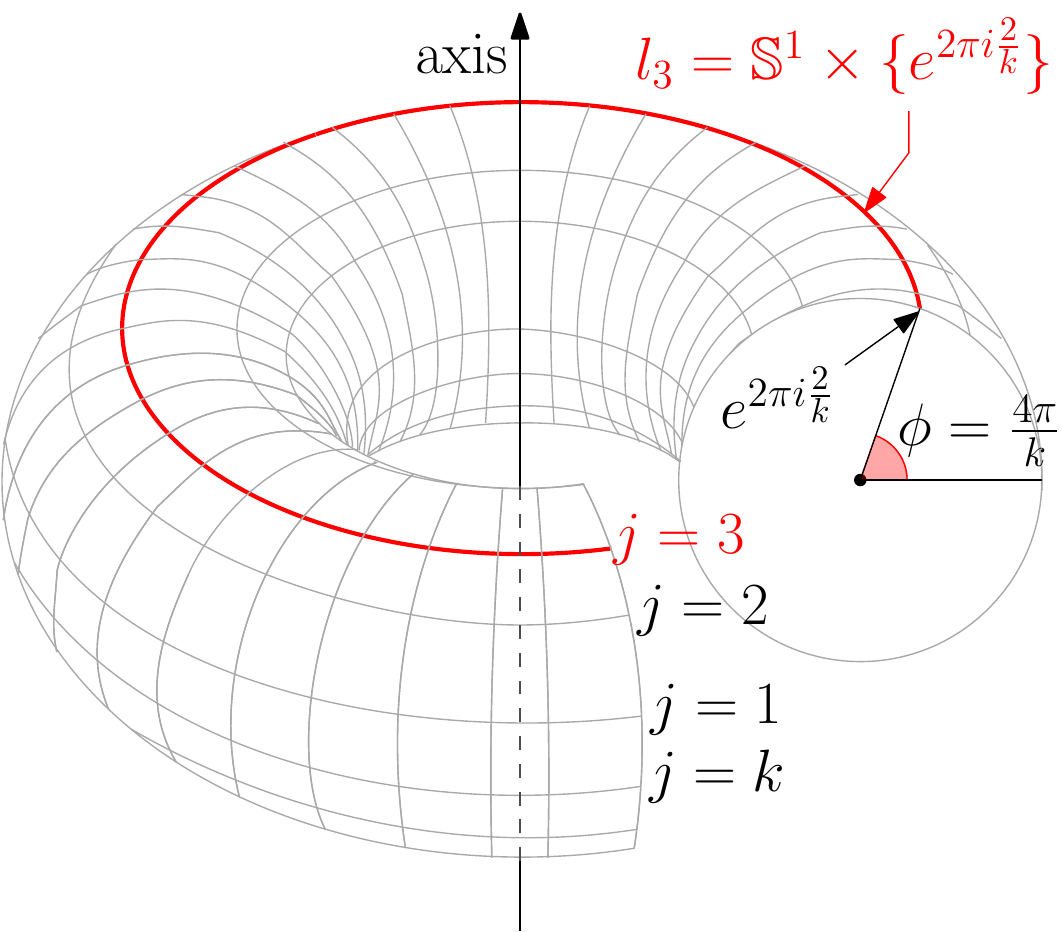}\label{fig:levels-T}}
\caption{Levels on (a) $\mathbb S$, on (b) $\mathbb R$, and on (c) $\mathbb T$, respectively.
} 
\label{fig:levels}
\end{figure}

Problems {\sc Radial}, {\sc Cyclic}, and {\sc Torus Level Planarity} take as input a level graph $G=(V,E,\gamma)$ and ask to find an embedding $\Gamma$ of $G$ on $\mathbb S$, on $\mathbb R$, and on $\mathbb T$, respectively, in which each vertex $v \in V$ lies on $l_{\gamma(v)}$ and each edge $(u,v) \in E$ is monotone. Embedding $\Gamma$ is called a {\em radial}, a {\em cyclic}, and a {\em torus level embedding} of $G$, respectively.
A level graph admitting a radial, cyclic, or torus level embedding is called {\em radial}, {\em cyclic}, or {\em torus level planar}, respectively.

Lemmata~\ref{le:radial-torus} and~\ref{le:cyclic-torus} show that the torus surface combines the power of representation of the standing and of the rolling cylinder. To strengthen this fact, we present a level graph in Fig.~\ref{fig:torus-level} that is neither radial nor cyclic level planar, yet it is torus level planar; note that the underlying (non-level) graph is also planar.

\begin{lemma}\label{le:radial-torus}
Every radial level planar graph is also torus level planar.
Further, {\sc Radial Level Planarity} reduces in linear time to {\sc Torus Level Planarity}.
\end{lemma}

\begin{proof}
The first part of the statement can be easily proved by observing that any level embedding on $\mathbb S$ is also a level embedding on $\mathbb T$.
We prove the second part of the statement. 
		Given an instance $G=(V=\bigcup_{i=1}^k{V_i},E,\gamma)$ of {\sc Radial Level Planarity} we construct an instance $G'=(V \cup \{a,b,c,d\}, E \cup \{(a,b),(b,c),(c,d),(d,a)\},\gamma')$ of {\sc Torus Level Planarity}, where $\gamma'(v) = \gamma(v)$, for each $v \in V$, $\gamma'(a)=\gamma'(c)=k$, and $\gamma'(b)=\gamma'(d)=1$. 
		Suppose that $G$ admits a radial level embedding $\Gamma$ on $\mathbb S$. Consider the corresponding torus level embedding $\Gamma'$ of $G$ on $\mathbb T$, which exists by the first part of the statement. Since $G$ does not contain any edge $(u,v) \in E$ such that $\gamma(u)>\gamma(v)$ (as otherwise, $G$ would not be radial level planar), the strip of $\mathbb T$ between $l_k$ and $l_1$ does not contain any edge. Hence, cycle $(a,b,c,d)$ can be added to $\Gamma'$ to obtain a torus level embedding of $G'$ on $\mathbb T$.
		Suppose that $G'$ admits a torus level embedding $\Gamma'$ on $\mathbb T$.
		A radial level embedding of $G$ on $\mathbb S$ can be obtained by removing the drawing of cycle $(a,b,c,d)$ in $\Gamma'$.
		\end{proof}

\begin{lemma}\label{le:cyclic-torus}
Every cyclic level planar graph is also torus level planar.
Further, {\sc Cyclic Level Planarity} reduces in linear time to {\sc Torus Level Planarity}.
\end{lemma}

\begin{proof}
The first part of the statement can be proved by observing that any level embedding on $\mathbb R$ is also a level embedding on $\mathbb T$.
We prove the second part. Given an instance $G=(V=\bigcup_{i=1}^k{V_i},E,\gamma)$ of {\sc Cyclic Level Planarity}, we construct an instance $G'=(V \cup \{w_1,\dots,w_k\},$ $E \cup \bigcup_{i=1}^{k-1}{(w_i,w_{i+1})} \cup \{(w_k,w_1)\},\gamma')$ of {\sc Torus Level Planarity}, where $\gamma'(v) = \gamma(v)$ for each $v \in V$, and $\gamma'(w_i)=i$.
Suppose that $G$ admits a cyclic level embedding $\Gamma$ on $\mathbb R$. Add a drawing of cycle $(w_1,\dots,w_k,w_1)$ to $\Gamma$ along the boundary of one of the two bases of $\mathbb R$, thus obtaining a cyclic level embedding $\Gamma'$ of $G'$ on $\mathbb R$. From the first part of the statement there exists a torus level embedding of $G'$ on $\mathbb T$.
Suppose that $G'$ admits a torus level embedding $\Gamma'$ on $\mathbb T$.
A cyclic level embedding of $G$ on $\mathbb R$ can be obtained by removing the drawing of cycle $(w_1,\dots,w_k,w_1)$ in~$\Gamma'$.
\end{proof}

\begin{figure}[tb!]
\centering
\subfloat[]{\includegraphics{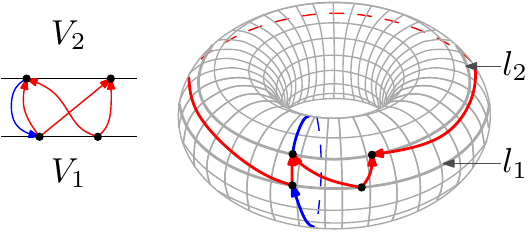}\label{fig:torus-level}}\hfil
\subfloat[]{\includegraphics{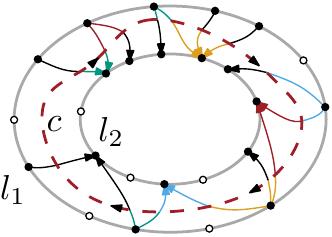}\label{fig:2-levels}}
\caption{(a) A level graph that is neither cyclic nor radial level planar, yet it is torus level planar. (b) A radial level embedding $\Gamma$ of a level graph $(V_1 \cup V_2,E,\gamma)$ on two levels. Colors are used for edges incident to vertices of degree larger than one to illustrate that the edge ordering on $E$ in $\Gamma$ is v-consecutive.} 
\label{fig:something}
\end{figure}

\smallskip
\noindent
{\bf Orderings and PQ-trees.} Let $A$ be a finite set. 
We call {\em linear ordering} any permutation of $A$.
When considering the first and the last elements of the permutation as consecutive, we talk about {\em circular orderings}. 
Let $\cal O$ be a circular ordering on $A$ and let ${\cal O}'$ be the circular ordering on $A' \subseteq A$ obtained by restricting $\cal O$ to the elements of $A'$. Then ${\cal O}'$ is a {\em suborder} of $\cal O$ and $\cal O$ is an {\em extension} of $\cal O'$. 
Let $A$ and $S$ be finite sets, let ${\cal O}' = s_1,s_2,\dots,s_{|S|}$ be a circular ordering on $S$, let  $\phi: S \rightarrow A$ be an injective map, and let $A' \subseteq A$ be the image of $S$ under $\phi$;
then $\phi({\cal O}')$ denotes the circular ordering $\phi(s_1),\phi(s_2),\dots,\phi(s_{|S|})$.
We also say that a circular ordering ${\cal O}'$ on $S$ is a {\em suborder} of a circular ordering $\cal O$ on $A$ (and $\cal O$ is an {\em extension} of ${\cal O}'$) if $\phi({\cal O'})$ is a suborder of ${\cal O}$.

An \emph{unrooted PQ-tree} $T$ is a tree whose leaves are the elements of a ground set $A$. PQ-tree $T$ can be used to represent all and only the circular orderings ${\cal O}(T)$ on $A$ satisfying a given set of \emph{consecutivity constraints} on $A$, each of which specifies that a subset of the elements of $A$ has to appear consecutively in all the represented circular orderings on $A$. 
The internal nodes of $T$ are either {\em P-nodes} or {\em Q-nodes}. The orderings in ${\cal O}(T)$ are all and only the circular orderings on the leaves of $T$ obtained by arbitrarily ordering the neighbours of each P-node and by arbitrarily selecting for each Q-node either a given circular ordering on its neighbours or its reverse ordering.
Note that possibly ${\cal O}(T)=\emptyset$, if $T$ is the empty tree, or ${\cal O}(T)$ represents all possible circular orderings on $A$, if $T$ is a star centered at a P-node. In the latter case, $T$ is the {\em universal} PQ-tree on $A$.

We illustrate three linear-time operations on PQ-trees (see~\cite{bl-tcop-76,hrl-tsp-13,hm-ppp-04}). Let $T$ and $T'$ be PQ-trees on $A$ and let $X\subseteq A$:
%
	The {\em reduction of $T$ by $X$} 
	builds a new PQ-tree on $A$ representing the circular orderings in ${\cal O}(T)$ in which the elements of $X$ are consecutive.
	The {\em projection of $T$ to $X$}, denoted by $T|_X$, builds a new PQ-tree on $X$ representing the circular orderings on $X$ that are suborders of circular orderings in ${\cal O}(T)$.
	The {\em intersection of $T$ and $T'$}, denoted by $T \cap T'$, builds a new PQ-tree on $A$ representing the circular orderings in ${\cal O}(T) \cap {\cal O}(T')$.

\smallskip
\noindent
{\bf Simultaneous PQ-Ordering.} Let $D=(N,Z)$ be a DAG with vertex set $N=\{T_1,\dots,T_k\}$, where $T_i$ is a PQ-tree, such that each arc $(T_i,T_j;\phi) \in Z$ consists of a source $T_i$, of a target $T_j$, and of an injective map $\phi: {\cal L}(T_j) \rightarrow {\cal L}(T_i)$ from the leaves of $T_j$ to the leaves of $T_i$. 
Given an arc $a=(T_i,T_j;\phi) \in Z$ and circular orderings ${\cal O}_i \in {\cal O}(T_i)$ and ${\cal O}_j \in {\cal O}(T_j)$, we say that arc $a$ is {\em satisfied} by $({\cal O}_i,{\cal O}_j)$ if ${\cal O}_i$ extends $\phi({\cal O}_j)$.
The {\sc Simultaneous PQ-Ordering} problem asks to find circular orderings ${\cal O}_1 \in {\cal O}(T_1),\dots,{\cal O}_k \in {\cal O}(T_k)$ on ${\cal L}(T_1),\dots,{\cal L}(T_k)$, respectively, such that each arc $(T_i,T_j;\phi) \in Z$ is satisfied by $({\cal O}_i,{\cal O}_j)$.

Let $(T_i,T_j;\phi)$ be an arc in $Z$. 
An internal node $\mu_i$ of $T_i$ is {\em fixed by} an internal node $\mu_j$ of $T_j$ (and $\mu_j$ {\em fixes} $\mu_i$ \emph{in} $T_i$) if there exist leaves $x,y,z \in {\cal L}(T_j)$ and $\phi(x),\phi(y),\phi(z) \in {\cal L}(T_i)$ such that 
(i) removing $\mu_j$ from $T_j$ makes $x$, $y$, and $z$ pairwise disconnected in $T_j$, and
(ii) removing $\mu_i$ from $T_i$ makes $\phi(x)$, $\phi(y)$, and $\phi(z)$ pairwise disconnected in $T_i$.
Note that by (i) the three paths connecting $\mu_j$ with $x$, $y$, and $z$ in $T_j$ share no node other than $\mu_j$, while by (ii) those connecting $\mu_i$ with $\phi(x)$, $\phi(y)$, and $\phi(z)$ in $T_i$ share no node other than $\mu_i$.
Since any ordering ${\cal O}_j$ determines a circular ordering around $\mu_j$ of the paths connecting it with $x$, $y$, and $z$ in $T_j$, any ordering ${\cal O}_i$ extending $\phi({\cal O}_j)$ determines the same circular ordering around $\mu_i$ of the paths connecting it with $\phi(x)$, $\phi(y)$, and $\phi(z)$ in $T_i$; this is why we say that $\mu_i$ is fixed by $\mu_j$.

Theorem~\ref{th:2-fixed} below will be a key ingredient in the algorithms of the next section. However, in order to exploit it, we need to consider {\em normalized} instances of {\sc Simultaneous PQ-Ordering}, namely instances $D=(N,Z)$ such that, for each arc $(T_i,T_j;\phi) \in Z$ and for each internal node $\mu_j \in T_j$, tree $T_i$ contains exactly one node $\mu_i$ that is fixed by $\mu_j$. This property can be guaranteed by an operation, called {\em normalization}~\cite{br-spacep-13}, defined as follows. Consider each arc $(T_i,T_j;\phi) \in Z$ and replace $T_j$ with $T_i|_{\phi({\cal L}(T_j))} \cap T_j$ in $D$, that is, replace tree $T_j$ with its intersection with the projection of its parent $T_i$ to the set of leaves of $T_i$ obtained by applying mapping $\phi$ to the leaves ${\cal L}(T_j)$ of $T_j$. 

Consider a normalized instance $D=(N,Z)$.
Let $\mu$ be a P-node of a PQ-tree $T$ with parents $T_1,\dots,T_r$ and let $\mu_i \in T_i$ be the unique node in $T_i$, with $1 \leq i \leq r$, fixed by $\mu$. The {\em fixedness} $fixed(\mu)$ of $\mu$ is defined as $fixed(\mu) = \omega + \sum^r_{i=1} (fixed(\mu_i)-1)$, where $\omega$ is the number of children of $T$ fixing $\mu$. A P-node $\mu$ is \emph{$k$-fixed} if $fixed(\mu) \leq k$. Also, instance $D$ is \emph{$k$-fixed} if all the P-nodes of any PQ-tree $T \in N$ are $k$-fixed. 

\begin{theorem}[Bl{\"{a}}sius and Rutter~\cite{br-spacep-13}, Theorems~3.2 and~3.3]\label{th:2-fixed} $2$-fixed instances of {\sc Simultaneous PQ-Ordering} can be tested in quadratic time.
\end{theorem}

\section{Torus Level Planarity} \label{se:cyclic}

In this section we provide a polynomial-time testing and embedding algorithm for {\sc Torus Level Planarity} that is based on the following simple observation.

\begin{observation}\label{prop:characterization}
A proper level graph $G=(\bigcup^k_{i=1}V_i,E,\gamma)$ is torus level planar if and only if 
there exist circular orderings ${\cal O}_1,\dots,{\cal O}_k$ on $V_1,\dots,V_k$
such that, for each $1 \leq i \leq k$ with $k+1=1$, there exists a radial level embedding of the level graph $(V_i \cup V_{i+1},(V_i \times V_{i+1}) \cap E,\gamma)$ in which the circular orderings on $V_i$ along $l_i$ and on $V_{i+1}$ along $l_{i+1}$ are ${\cal O}_i$ and ${\cal O}_{i+1}$, respectively.
\end{observation}

In view of Observation~\ref{prop:characterization} we focus on a level graph $G=(V_1 \cup V_2,E,\gamma)$ on two levels $l_1$ and $l_2$. Denote by $V_1^+$ and by $V_{2}^-$ the subsets of $V_1$ and of $V_2$ that are incident to edges in $E$, respectively.
Let $\Gamma$ be a radial level embedding of $G$.
Consider a closed curve $c$ separating levels $l_1$ and $l_2$ and intersecting all the edges in $E$ exactly once.
The {\em edge ordering on $E$ in $\Gamma$} is the circular ordering in which the edges in $E$ intersect $c$ according to a clockwise orientation of $c$ on the surface~$\mathbb S$ of the standing cylinder; see Fig.~\ref{fig:2-levels}.
Further, let $\cal O$ be a circular ordering on the edge set $E$. Ordering $\cal O$ is {\em vertex-consecutive} ({\em v-consecutive}) if all the edges incident to each vertex in $V_1 \cup V_2$ are consecutive in $\cal O$. 

Let $\cal O$ be a v-consecutive ordering on $E$. 
We define orderings ${\cal O}^+_1$ on $V_1^+$ and ${\cal O}^-_2$ on $V_2^-$ {\em induced by} ${\cal O}$, as follows.
Consider the edges in $E$ one by one as they appear in ${\cal O}$. Append the end-vertex in $V^+_1$ of the currently considered edge to a list ${\cal L}^+_1$. 
Since $\cal O$ is v-consecutive, the occurrences of the same vertex appear consecutively in ${\cal L}^+_1$, regarding such a list as circular.
Hence, ${\cal L}^+_1$ can be turned into a circular ordering ${\cal O}^+_1$ on $V^+_1$ by removing repetitions of the same vertex. Circular ordering ${\cal O}^-_2$ can be constructed analogously. We have the following.

\begin{lemma}\label{le:torus-characterization}
Let $\cal O$ be a circular ordering on $E$ and $({\cal O}_1,{\cal O}_2)$ be a pair of circular orderings on $V_1$ and $V_2$. 
There exists a radial level embedding of $G$ whose edge ordering is $\cal O$ and such that the circular orderings on $V_1$ and $V_2$ along $l_1$ and $l_2$ are ${\cal O}_1$ and ${\cal O}_2$, respectively,
if and only if $\cal O$ is v-consecutive, and ${\cal O}_1$ and ${\cal O}_2$ extend the orderings ${\cal O}^+_1$ and ${\cal O}^-_2$ on $V_1^+$ and $V_2^-$ induced by $\cal O$, respectively.
\end{lemma}

\begin{proof}
The necessity is trivial. 
For the sufficiency, assume that $\cal O$ is v-consecutive and that ${\cal O}_1$ and ${\cal O}_2$ extend the orderings ${\cal O}^+_1$ and ${\cal O}^-_2$ on $V_1^+$ and $V_2^-$ induced by $\cal O$, respectively.
We construct a radial level embedding $\Gamma$ of $G$ with the desired properties, as follows. Let $\Gamma^*$ be a radial level embedding consisting of $|E|$ non-crossing curves, each connecting a distinct point on $l_1$ and a distinct point on $l_2$. We associate each curve with a distinct edge in $E$, so that the edge ordering of $\Gamma^*$ is $\cal O$. Note that, since $\cal O$ is v-consecutive, all the occurrences of the same vertex of $V^+_1$ and of $V_2^-$ appear consecutively along $l_1$ and $l_2$, respectively. Hence, we can transform $\Gamma^*$ into a radial level embedding $\Gamma'$ of
$G'=(V_1^+ \cup V_2^-, E, \gamma)$, by continuously deforming the curves in $\Gamma^*$ incident to occurrences of the same vertex in $V_1^+$ (in $V_2^-$) so that their end-points on $l_1$ (on $l_2$) coincide. Since the circular orderings on $V_1^+$ and on $V_2^-$ along $l_1$ and $l_2$ are ${\cal O}^+_1$ and ${\cal O}^-_2$, respectively, we can construct $\Gamma$ by inserting the isolated vertices in $V_1 \setminus V_1^+$ and $V_2 \setminus V_2^-$ at suitable points along $l_1$ and $l_2$, so that the circular orderings on $V_1$ and on $V_2$ along $l_1$ and $l_2$ are ${\cal O}_1$ and ${\cal O}_2$, respectively.
\end{proof}

We construct an instance $I(G)$ of {\sc Simultaneous PQ-Ordering} starting from a level graph $G=(V_1 \cup V_2,E,\gamma)$ on two levels as follows; refer to Fig.~\ref{fig:simpq-instance}, where $I(G)$ corresponds to the subinstance $I(G_{i,i+1})$ in the dashed box.
We define the {\em level trees} $T_1$ and $T_2$ as the universal PQ-trees on $V_1$ and $V_2$, respectively. 
Also, we define the {\em layer tree} $T_{1,2}$ as the PQ-tree on $E$ representing exactly the edge orderings for which a radial level embedding of $G$ exists, which are the v-consecutive orderings on $E$, by Lemma~\ref{le:torus-characterization}. 
The PQ-tree $T_{1,2}$ can be constructed in $O(|G|)$ time~\cite{bl-tcop-76,hm-ppp-04}.
We define the {\em consistency trees} $T_1^+$ and $T_2^-$ as the universal PQ-trees on $V_1^+$ and $V_2^-$, respectively.
Instance $I(G)$ contains $T_1$, $T_2$, $T_{1,2}$, $T_1^+$, and $T_2^-$, together with the arcs $(T_1, T_1^+, \iota)$,
$(T_2, T_2^-, \iota)$,
$(T_{1,2}, T_1^+, \phi_1^+)$,
 and $(T_{1,2}, T_2^-, \phi_{2}^-)$, where $\iota$ denotes the identity map and $\phi_1^+$  ($\phi_2^-$) assigns to each vertex in $V_1^+$ (in $V_2^-$) an incident edge in $E$. We have the following.

\begin{figure}[tb]
\centering
\includegraphics[width=\textwidth,page=2]{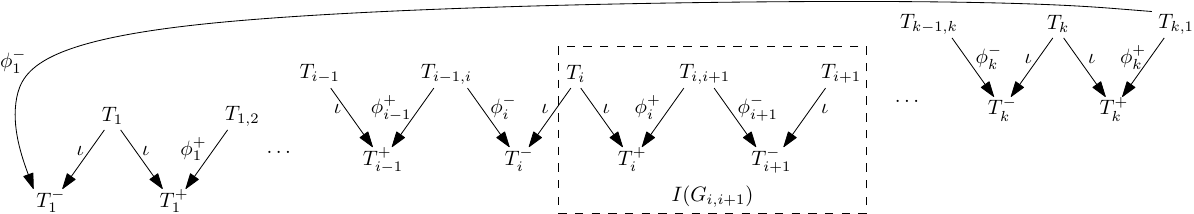}
\caption{Instance $I^*(G)$ of {\sc Simultaneous PQ-Ordering} for a level graph $G=(V,E,\gamma)$. Instance $I(G_{i,i+1})$ corresponding to the level graph $(V_i \cup V_{i+1}, (V_i \times V_{i+1}) \cap E, \gamma)$ induced by levels $i$ and $i+1$ of $G$ is enclosed in a dashed box.}
\label{fig:simpq-instance}
\end{figure}

\begin{lemma}\label{le:equivalence}
Level graph $G$ admits a radial level embedding in which the circular ordering on $V_1$ along $l_1$ is ${\cal O}_1$ and the circular ordering on $V_2$ along $l_2$ is ${\cal O}_2$ if and only if instance $I(G)$ of {\sc Simultaneous PQ-Ordering} admits a solution in which the circular ordering on $\mathcal{L}(T_1)$ is ${\cal O}_1$ and the one on $\mathcal{L}(T_2)$ is ${\cal O}_2$.
\end{lemma}

\begin{proof}
We prove the necessity. Let $\Gamma$ be a radial level embedding of $G$. We construct an ordering on the leaves of each tree in $I(G)$ as follows. Let ${\cal O}_1$, ${\cal O}_2$, ${\cal O}^+_1$, and ${\cal O}^-_2$ be the circular orderings on $V_1$ along $l_1$, on $V_2$ along $l_2$, on $V_1^+$ along $l_1$, and on $V_2^-$ along $l_2$ in $\Gamma$, respectively. Let ${\cal O}$ be the edge ordering on $E$ in $\Gamma$. Note that ${\cal O}\in {\cal O}(T_{1,2})$ since ${\cal O}$ is v-consecutive by Lemma~\ref{le:torus-characterization}. The remaining trees are universal, hence ${\cal O}_{1}\in {\cal O}(T_{1})$, ${\cal O}_{2}\in {\cal O}(T_{2})$, ${\cal O}^+_{1}\in {\cal O}(T^+_{1})$, and ${\cal O}^-_{2}\in {\cal O}(T^-_{2})$. 

We prove that all arcs of $I(G)$ are satisfied. Arc $(T_1,T_1^+,\iota)$ is satisfied if and only if ${\cal O}_1$ extends ${\cal O}^+_1$. This is the case since $\iota$ is the identity map, since $V_1^+ \subseteq V_1$, and since ${\cal O}_1$ and ${\cal O}^+_1$ are the circular orderings on $V_1$ and $V_1^+$ along $l_1$. Analogously, arc $(T_2,T_2^-,\iota)$ is satisfied. Arc $(T_{1,2},T_1^+,\phi^+_1)$ is satisfied if and only if ${\cal O}$ extends ${\cal O}^+_1$. This is due to the fact that $\phi^+_1$ assigns to each vertex in $V_1^+$ an incident edge in $E$ and to the fact that, by Lemma~\ref{le:torus-characterization}, ordering ${\cal O}$ is v-consecutive and ${\cal O}^+_1$ is induced by ${\cal O}$. Analogously, arc $(T_{1,2},T_{2}^-,\phi^-_2)$ is satisfied.  

We prove the sufficiency. Suppose that $I(G)$ is a positive instance of {\sc Simultaneous PQ-Ordering}, that is, there exist orderings ${\cal O}_1$, ${\cal O}_2$, ${\cal O}^+_1$, ${\cal O}^-_2$, and ${\cal O}$ of the leaves of the trees $T_1$, $T_2$, $T^+_1$, $T^-_2$, and $T_{1,2}$, respectively, satisfying all arcs of $I(G)$. 
Since $\iota$ is the identity map and since arcs $(T_1,T_1^+,\iota)$ and $(T_2,T_2^-,\iota)$ are satisfied, we have that ${\cal O}^+_1$ and ${\cal O}^-_2$ are restrictions of ${\cal O}_1$ and ${\cal O}_2$ to $V_1^+$ and $V_2^-$, respectively. Also, since $(T_{1,2},T_1^+,\phi^+_1)$ and $(T_{1,2},T_{2}^-,\phi^-_2)$ are satisfied, we have that ${\cal O}$ extends both ${\cal O}^+_1$ and ${\cal O}^-_2$. By the construction of $T_{1,2}$, ordering ${\cal O}$ is v-consecutive. By Lemma~\ref{le:torus-characterization}, a radial level embedding $\Gamma$ of $G$ exists in which the circular ordering on $V_i$ along $l_i$ is ${\cal O}_i$, for $i=1,2$.
\end{proof}

We now show how to construct an instance $I^*(G)$ of {\sc Simultaneous PQ-Ordering} from a proper level graph $G=(\bigcup^k_{i=1} V_i,E,\gamma)$ on $k$ levels; refer to Fig.~\ref{fig:simpq-instance}. For each $i=1,\dots,k$, let $I(G_{i,i+1})$ be the instance of {\sc
  Simultaneous PQ-Ordering} constructed as described above starting from the level graph on two levels $G_{i,i+1}=(V_i \cup V_{i+1}, (V_i \times V_{i+1}) \cap E, \gamma)$ (in the construction $V_i$ takes the role of $V_1$, $V_{i+1}$ takes the role of $V_2$, and  $k+1=1$). Any two instances $I(G_{i-1,i})$ and $I(G_{i,i+1})$ share exactly the level tree $T_i$, whereas non-adjacent instances are disjoint. We define $I^\cup(G)= \bigcup^k_{i=1} I(G_{i,i+1})$ and obtain $I^*(G)$ by normalizing $I^\cup(G)$.  We now present two lemmata about properties of instance $I^*(G)$. 

\begin{lemma} \label{le:construction}
$I^*(G)$ is $2$-fixed, has $O(|G|)$ size, and can be built in $O(|G|)$ time. 
\end{lemma} 

\begin{proof}
  Every PQ-tree $T$ in $I^\cup(G)$ is either a source with exactly two children or a sink with exactly two parents, and the normalization of $I^\cup(G)$ to obtain $I^*(G)$ does not alter this property. Thus every P-node in a PQ-tree $T$ in $I^*(G)$ is at most $2$-fixed.
  In fact, recall that for a P-node $\mu$ of a PQ-tree $T$ with parents $T_1,\dots,T_r$, we have that $fixed(\mu) = \omega + \sum^r_{i=1} (fixed(\mu_i)-1)$, where $\omega$ is the number
of children of $T$ fixing $\mu$, and  $\mu_i \in T_i$ is the unique node in $T_i$, with $1 \leq i \leq r$, fixed by $\mu$. Hence, if $T$ is a source PQ-tree, it holds $\omega = 2$ and $r=0$; whereas, if $T$ is a sink PQ-tree, it holds $\omega = 0$, $r=2$, and $fixed(\mu_i)=2$ for each parent $T_i$ of $T$. Therefore $I^*(G)$ is $2$-fixed.

Since every internal node of a PQ-tree in $I^*(G)$ has degree greater than $2$, to prove the bound on $|I^*(G)|$ it suffices to show that the total number of leaves of all PQ-trees in $I^*(G)$ is in $O(|G|)$. Since ${\cal L}(T_i)=V_i$ and ${\cal L}(T^-_i),{\cal L}(T^+_i)\subseteq
V_i$, the number of leaves of all level and consistency trees is at most $3\sum_{i=1}^k |V_i|\in O(|G|)$. Also, since ${\cal L}(T_{i,i+1})= (V_i \times V_{i+1}) \cap E$, the number of leaves of all layer trees is at most $\sum_{i=1}^k |(V_i \times V_{i+1}) \cap E|\in O(|G|)$. Thus $|I^*(G)|\in O(|G|)$.

We have already observed that each layer tree $T_{i,i+1}$ can be constructed in $O(|G_{i,i+1}|)$ time; level and consistency trees are stars, hence they can be constructed in linear time in the number of their leaves. Finally, the normalization of each arc $(T_i,T_j; \phi)$ can be performed in $O(|T_i|+|T_j|)$ time~\cite{br-spacep-13}. Hence, the $O(|G|)$ time bound follows.  
\end{proof}

\begin{lemma} \label{le:reduction}
Level graph $G$ admits a torus level embedding if and only if $I^*(G)$ is a positive instance of {\sc Simultaneous PQ-Ordering}.
\end{lemma}

\begin{proof}
Suppose that $G$ admits a torus level embedding $\Gamma$. For $i=1,\dots,k$, let ${\cal O}_i$ be the circular ordering on $V_i$ along $l_i$. By Observation~\ref{prop:characterization}, embedding $\Gamma$ determines a radial level embedding $\Gamma_{i,i+1}$ of $G_{i,i+1}$. By Lemma~\ref{le:equivalence}, for $i=1,\dots,k$, there exists a solution for the instance $I(G_{i,i+1})$ of {\sc Simultaneous PQ-Ordering} in which the circular ordering on ${\cal L}(T_i)$ (${\cal L}(T_{i+1})$) is ${\cal O}_i$ (resp.\ ${\cal O}_{i+1}$). Since the circular ordering on ${\cal L}(T_i)$ is ${\cal O}_i$ both in $I(G_{i-1,i})$ and $I(G_{i,i+1})$ and since each arc of $I^*(G)$ is satisfied as it belongs to exactly one instance $I(G_{i,i+1})$, which is a positive instance of {\sc Simultaneous PQ-Ordering}, it follows that the circular orderings deriving from instances $I(G_{i,i+1})$ define a solution for $I^*(G)$. 

Suppose that $I^*(G)$ admits a solution. Let ${\cal O}_1,\dots,{\cal O}_k$ be the circular orderings on the leaves of the level trees $T_1,\dots,T_k$ in this solution. By Lemma~\ref{le:equivalence}, for each $i=1,\dots, k$ with $k+1=1$, there exists a radial level embedding of level graph
$G_{i,i+1}$ in which the circular orderings on $V_i$ along $l_i$ and $V_{i+1}$ along $l_{i+1}$ are ${\cal O}_{i}$ and ${\cal O}_{i+1}$, respectively. By Observation~\ref{prop:characterization}, $G$ is torus level planar. 
\end{proof}

We thus get the main result of this paper.

\begin{theorem}\label{th:torus}
{\sc Torus Level Planarity} can be tested in quadratic (quartic) time for proper (non-proper) instances.
\end{theorem}

\begin{proof}
Consider any instance $G$ of {\sc Torus Level Planarity}. Assume first that $G$ is proper. By Lemmata~\ref{le:construction} and~\ref{le:reduction}, a $2$-fixed instance $I^*(G)$ of {\sc Simultaneous PQ-Ordering} equivalent to $G$ can be constructed in linear time with $|I^*(G)| \in O(|G|)$. By Theorem~\ref{th:2-fixed} instance $I^*(G)$ can be tested in quadratic time.

If $G$ is not proper, then subdivide every edge $(u,v)$ that spans $h>2$ levels with $h-2$ vertices, assigned to levels $\gamma(u)+1,\gamma(u)+2,\dots,\gamma(v)-1$. This increases the size of the graph at most quadratically, and the time bound follows.
\end{proof}

Theorem~\ref{th:torus} and Lemma~\ref{le:cyclic-torus} imply the following result.

\begin{theorem}\label{co:cyclicLevel-polynomial}
{\sc Cyclic Level Planarity} can be solved in quadratic (quartic) time for proper (non-proper) instances.
\end{theorem}


Our techniques allow us to solve a more general problem, that we call {\sc Torus $\cal T$-Level Planarity}, in which a level graph $G=(\bigcup^k_{i=1} V_i,E,\gamma)$ is given together with a set of PQ-trees ${\cal T}=\{\overline{T}_1,\dots,\overline{T}_k\}$ such that ${\cal L}(\overline{T}_i)= V_i$, where each tree $\overline{T}_i$ encodes consecutivity constraints on the ordering on $V_i$ along $l_i$. The goal is then to test the existence of a level embedding of $G$ on $\mathbb T$ in which the circular ordering on $V_i$ along $l_i$ belongs to ${\cal O} (\overline{T}_i)$. 
This problem has been studied in the plane~\cite{tibp-addfr-15,wsp-gktlg-12} under the name of {\sc $\cal T$-Level Planarity}; it is NP-hard in general and polynomial-time solvable for proper instances. While the former result implies the NP-hardness of {\sc Torus $\cal T$-Level Planarity}, the techniques of this paper show that {\sc Torus $\cal T$-Level Planarity} can be solved in polynomial time for proper instances. Namely, in the construction of instance $I^*(G)$ of {\sc Simultaneous PQ-Ordering}, it suffices to replace level tree $T_i$ with PQ-tree $\overline{T}_i$. Analogous considerations allow us to extend this result to {\sc Radial $\cal T$-Level Planarity} and {\sc Cyclic $\cal T$-Level Planarity}.

\section{Simultaneous Level Planarity} \label{se:simultaneous}

In this section we prove that  {\sc Simultaneous Level Planarity} is \NPC for two graphs on three levels and for three graphs on two levels, while it is polynomial-time solvable for two graphs on two levels. 

Both NP-hardness proofs rely on a reduction from the \NPC problem {\sc Betweenness}~\cite{o-top-79}, which asks for a ground set $S$ and a set $X$ of ordered triplets of $S$, with $|S|=n$ and $|X|=k$, whether a linear order $\prec$ of $S$ exists such that, for any $(\alpha,\beta,\gamma) \in X$, it holds true that $\alpha \prec \beta \prec \gamma$ or that $\gamma \prec \beta \prec \alpha$. Both proofs exploit the following gadgets. 

The {\em ordering gadget} is a pair \sefeinstance{} of level graphs on levels $l_1$ and $l_2$, where $l_1$ contains $nk$ vertices $u_{i,j}$, with $i=1,\dots,k$ and $j=1,\dots,n$, and $l_2$ contains $n (k-1)$ vertices $v_{i,j}$, with $i=1,\dots,k-1$ and $j=1,\dots,n$. For $i=1,\dots,k-1$ and $j=1,\dots,n$, \Gr{} contains edge $(u_{i,j},v_{i,j})$ and \Gb{} contains edge $(u_{i+1,j},v_{i,j})$. See \Gr{} and \Gb{} in Fig.~\ref{fig:simultaneous-level}(a). Consider any simultaneous level embedding $\Gamma$ of \sefeinstance{} and assume, w.l.o.g. up to a renaming, that $u_{1,1},\dots,u_{1,n}$ appear in this left-to-right order along $l_1$.

\begin{lemma} \label{le:sl-ordering}
For every $i=1,\dots,k$, vertices $u_{i,1},\dots,u_{i,n}$ appear in this left-to-right order along $l_1$ in~$\Gamma$; also, for every $i=1,\dots,k-1$, vertices $v_{i,1},\dots,v_{i,n}$ appear in this left-to-right order along~$l_2$ in~$\Gamma$. 
\end{lemma} 

\begin{proof}
	Suppose, for a contradiction, that the statement does not hold. Then let $k^*$ be the smallest index such that either: 
	
	\begin{itemize}
		\item[(A)] for every $i=1,\dots,k^*-1$, vertices $u_{i,1},\dots,u_{i,n}$ appear in this left-to-right order along $l_1$; for every $i=1,\dots,k^*-1$, vertices $v_{i,1},\dots,v_{i,n}$ appear in this left-to-right order along $l_2$; and vertices $u_{k^*,1},\dots,u_{k^*,n}$ do not appear in this left-to-right order along $l_1$; or 
		\item[(B)] for every $i=1,\dots,k^*$, vertices $u_{i,1},\dots,u_{i,n}$ appear in this left-to-right order along $l_1$; for every $i=1,\dots,k^*-1$, vertices $v_{i,1},\dots,v_{i,n}$ appear in this left-to-right order along $l_2$; and vertices $v_{k^*,1},\dots,v_{k^*,n}$ do not appear in this left-to-right order along $l_2$.
	\end{itemize}
	
	Suppose that we are in Case (A), as the discussion for Case (B) is analogous. Then $v_{k^*-1,1},\dots,v_{k^*-1,n}$ appear in this left-to-right order along $l_2$, while $u_{k^*,1},\dots,u_{k^*,n}$ do not appear in this left-to-right order along $l_1$. Hence, there exist indices $i$ and $j$ such that $v_{k^*-1,i}$ is to the left of $v_{k^*-1,j}$ along $l_2$, while $u_{k^*,i}$ is to the right of $u_{k^*,j}$ along $l_1$. Hence, edges $(u_{k^*,i},v_{k^*-1,i})$ and  $(u_{k^*,j},v_{k^*-1,j})$ cross, thus contradicting the assumption that $\Gamma$ is a simultaneous level embedding, as they both belong to \Gb{}. 
\end{proof}

The {\em triplet gadget} is a path $T=(w_1,\dots,w_5)$ on two levels, where $w_1$, $w_3$, and $w_5$ belong to a level $l_i$ and $w_2$ and $w_4$ belong to a level $l_j\neq l_i$. See $\green{G_3}$ in Fig.~\ref{fig:simultaneous-level}(a), with $i=1$ and $j=2$. We have the following.

\begin{lemma} \label{le:sl-triplet}
In every level embedding of $T$, vertex $w_3$ is between $w_1$ and $w_5$ along~$l_i$.
\end{lemma} 

\begin{proof}
	Suppose, for a contradiction, that $w_3$ is to the left of $w_1$ and $w_5$ along $l_i$; the case in which it is to their right is analogous. Also assume that $l_i$ is below $l_j$, as the other case is symmetric. If $w_2$ is to the left of $w_4$ along $l_j$, then edges $(w_1,w_2)$ and $(w_3,w_4)$ cross, otherwise edges $(w_3,w_2)$ and $(w_5,w_4)$ cross. In both cases we have a contradiction.
\end{proof}

We are now ready to prove the claimed {\cal NP}-completeness results.

\begin{figure}[tb]
\centering
\subfloat[]{\includegraphics[scale=1]{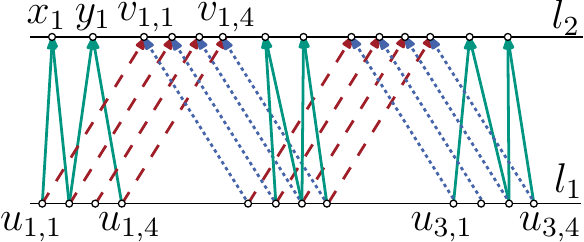}} 
\subfloat[]{\includegraphics[scale=1]{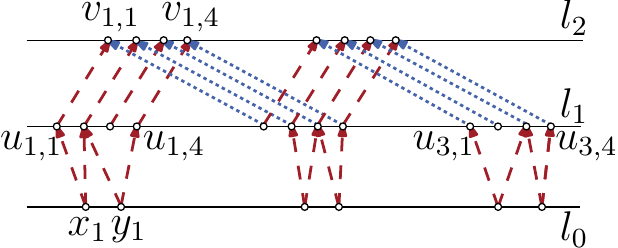}} 
\caption{Instances (a) $\langle \red{G_1},\blue{G_2},\green{G_3} \rangle$ and (b) \sefeinstance{} corresponding to an instance of {\sc Betweenness} with $X=\{(u_{1,1},u_{1,2},u_{1,4}),(u_{1,2},u_{1,3},u_{1,4}),(u_{1,1},u_{1,3},u_{1,4})\}$.} 
\label{fig:simultaneous-level}
\end{figure}

\begin{theorem} \label{thm:sim-level-npc-two-levels} {\sc Simultaneous Level Planarity} is \NPC even for three graphs on two levels and for two graphs on three levels.
\end{theorem}
\begin{proof}
Both problems clearly are in \NP. We prove the \NP-hardness only for three graphs on two levels (see Fig.~\ref{fig:simultaneous-level}(a)), as the other proof is analogous (see Fig.~\ref{fig:simultaneous-level}(b)). We construct an instance $\langle \red{G_1(V,E_1,\gamma)},\blue{G_2(V,E_2,\gamma)},\green{G_3(V,E_3,\gamma)} \rangle$ of {\sc Simultaneous Level Planarity} from an instance $(S=\{u_{1,1},\dots,u_{1,n}\},X=\{(u_{1,a_i},u_{1,b_i},u_{1,c_i}):i=1,\dots,k\})$ of {\sc Betweenness} as follows: Pair \sefeinstance{} contains an ordering gadget on levels $l_1$ and $l_2$, where the vertices $u_{1,1},\dots,u_{1,n}$ of \Gr{} are (in bijection with) the elements of $S$. Graph \Gg{} contains $k$ triplet gadgets $T_i(u_{i,a_i},x_i,u_{i,b_i},y_i,u_{i,c_i})$, for $i=1,\dots,k$. Vertices $x_1,y_1,\dots,x_k,y_k$ are all distinct and are on $l_2$. Clearly, the construction can be carried out in linear time.  We prove the equivalence of the two instances.

$(\Longrightarrow)$ Suppose that a simultaneous level embedding $\Gamma$ of \sefethreeinstance{} exists. We claim that the left-to-right order of $u_{1,1},\dots,u_{1,n}$ along $l_1$ satisfies the betweenness constraints in $X$. Suppose, for a contradiction, that there exists a triplet $(u_{1,a_i},u_{1,b_i},u_{1,c_i})\in X$ with $u_{1,b_i}$ not between $u_{1,a_i}$ and $u_{1,c_i}$ along $l_1$. By Lemma~\ref{le:sl-ordering}, $u_{i,b_i}$ is not between $u_{i,a_i}$ and $u_{i,c_i}$. By Lemma~\ref{le:sl-triplet}, we have that $T_i(u_{i,a_i},x_i,u_{i,b_i},y_i,u_{i,c_i})$ is not planar in  $\Gamma$, a contradiction. 

$(\Longleftarrow)$ Suppose that $(S,X)$ is a positive instance of {\sc Betweenness}, and assume, w.l.o.g. up to a renaming, that $\prec:=u_{1,1},\dots,u_{1,n}$ is a solution for $(S,X)$.
Construct a straight-line simultaneous level planar drawing of \sefethreeinstance{} with:
\begin{inparaenum}[(i)]
\item vertices $u_{1,1},\dots,u_{1,n},$ $\dots, u_{k,1},$ $\dots,u_{k,n}$ in this
  left-to-right order along $l_1$, 
\item  vertices $v_{1,1},\dots,v_{1,n},\dots,v_{k-1,1},$ $\dots,v_{k-1,n}$ in this
  left-to-right order along $l_2$,
\item  vertices $x_i$ and $y_i$ to the left
  of vertices $x_{i+1}$ and $y_{i+1}$, for $i=1,\dots,k-1$, and
\item vertex $x_i$ to the left of vertex $y_i$ if and only if $u_{1,a_i}\prec u_{1,c_i}$.
\end{inparaenum}

Properties (i) and (ii) guarantee that, for any two edges $(u_{i,j},v_{i,j})$ and $(u_{i',j'},v_{i',j'})$, vertex $u_{i,j}$ is to the left of $u_{i',j'}$ along $l_1$ if and only if $v_{i,j}$ is to the left of $v_{i',j'}$ along $l_2$, which implies the planarity of \Gr{} in $\Gamma$. The planarity of \Gb{} in $\Gamma$ is proved analogously. Properties (i) and (iii) imply that no two paths $T_i$ and $T_j$ cross each other, while Property (iv) guarantees that each path $T_i$ is planar. Hence, the drawing of \Gg{} in $\Gamma$ is planar.
\end{proof}

The graphs in Theorem~\ref{thm:sim-level-npc-two-levels} can be made connected, by adding vertices and edges, at the expense of using one additional level. Also, the theorem holds true even if the simultaneous embedding is {\em geometric} or {\em with fixed edges} (see~\cite{bkr-sepg-13,bcdeeiklm-spge-07} for definitions).

In contrast to the NP-hardness results, a reduction to a proper instance of {\sc Cyclic Level Planarity} allows us to decide in polynomial time instances composed of two graphs on two levels. Namely, the edges of a graph are directed from $l_1$ to $l_2$, while those of the other graph are directed from $l_2$ to $l_1$.

\begin{theorem} \label{th:sim-poly}
{\sc Simultaneous Level Planarity} is quadratic-time solvable for two graphs on two levels. 
\end{theorem}

\begin{proof}
Let $\langle \red{G_1(V,E_1,\gamma)},\blue{G_2(V,E_2,\gamma)} \rangle$ be an instance of the {\sc Simultaneous Level Planarity} problem, where each of \Gr{} and \Gb{} is a level graph on two levels $l_1$ and $l_2$. We define a proper instance $(V,E,\gamma)$ of {\sc Cyclic Level Planarity} as follows. The vertex set $V$ is the same as the one of \Gr{} and \Gb{}, as well as the function $\gamma:V\rightarrow\{1,2\}$; further, $E$ contains an edge $(u,v)$ for every $(u,v) \in E_1$ and an edge $(v,u)$ for every $(u,v) \in E_2$. We prove that $\langle \red{G_1(V,E_1,\gamma)},\blue{G_2(V,E_2,\gamma)} \rangle$ is simultaneous level planar if and only if $(V,E,\gamma)$ is cyclic level planar.

$(\Longrightarrow)$ Consider a simultaneous level embedding of \sefeinstance{}, map it to the surface $\mathbb R$ of the rolling cylinder, and wrap the edges of \Gb{} around the part delimited by $l_1$ and $l_2$ not containing the edges of \Gr{}, hence obtaining a cyclic level embedding of $(V,E,\gamma)$. 

$(\Longleftarrow)$ Consider a cyclic level embedding of $(V,E,\gamma)$ on $\mathbb R$, reroute the edges of \Gb{} to lie in the part of $\mathbb R$ delimited by $l_1$ and $l_2$ and containing the edges of \Gr{}, and map this drawing to the plane; this results in a simultaneous level embedding of \Gr{} and \Gb{}. 

The statement of the theorem then follows from Corollary~\ref{co:cyclicLevel-polynomial} and from the fact that the described reduction can be performed in linear time.
\end{proof}





\section{Conclusions and Open Problems} \label{se:conclusions}

In this paper we have settled the computational complexity of two of the main open problems in the research topic of level planarity by showing that the {\sc Cyclic Level Planarity} and the {\sc Torus Level Planarity} problems are polynomial-time solvable.
%
%
Our algorithms run in quartic time in the graph size; it is hence an interesting challenge to design new techniques to improve this time bound. 
We also introduced a notion of simultaneous level planarity for level graphs and we established a complexity dichotomy for this problem.

\begin{wrapfigure}[5]{r}{0.53\textwidth}
  \vspace{-20pt}
  \centering
\includegraphics{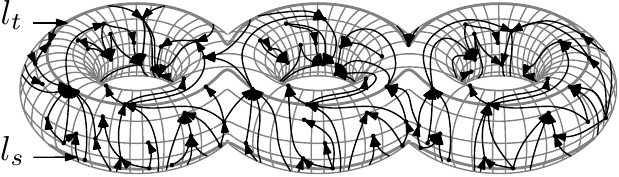}
\label{fig:double-torus}
\end{wrapfigure}

An intriguing research direction is the one of extending the concept of level planarity to surfaces with genus larger than one. 
However, there seems to be more than one meaningful way to arrange $k$ levels on a high-genus surface.
A reasonable choice would be the one shown in the figure, in which the levels are arranged in different sequences between two distinguished levels $l_s$ and $l_t$ (and edges only connect vertices on two levels in the same sequence). 
{\sc Radial Level Planarity} and {\sc Torus Level Planarity} can be regarded as special cases of this setting (with only one and two paths of levels between $l_s$ and $l_t$, respectively). 


\bibliographystyle{splncs03}
\bibliography{bibliography}

\begin{thebibliography}{10}
\providecommand{\url}[1]{\texttt{#1}}
\providecommand{\urlprefix}{URL }

\bibitem{tibp-addfr-15}
Angelini, P., {Da Lozzo}, G., {Di Battista}, G., Frati, F., Roselli, V.: The
  importance of being proper: (in clustered-level planarity and {T}-level
  planarity). Theor. Comput. Sci.  571,  1--9 (2015),
  \url{http://dx.doi.org/10.1016/j.tcs.2014.12.019}

\bibitem{abbg-cpue-12}
Auer, C., Bachmaier, C., Brandenburg, F., Glei{\ss}ner, A.: Classification of
  planar upward embedding. In: van Kreveld, M.J., Speckmann, B. (eds.) Graph
  Drawing - 19th International Symposium, {GD} 2011, Eindhoven, The
  Netherlands, September 21-23, 2011, Revised Selected Papers. LNCS, vol. 7034,
  pp. 415--426. Springer (2012),
  \url{http://dx.doi.org/10.1007/978-3-642-25878-7_39}

\bibitem{bbf-rlptelt-05}
Bachmaier, C., Brandenburg, F., Forster, M.: Radial level planarity testing and
  embedding in linear time. J. Graph Algorithms Appl.  9(1),  53--97 (2005),
  \url{http://jgaa.info/accepted/2005/BachmaierBrandenburgForster2005.9.1.pdf}

\bibitem{bb-ltpte-08}
Bachmaier, C., Brunner, W.: Linear time planarity testing and embedding of
  strongly connected cyclic level graphs. In: Halperin, D., Mehlhorn, K. (eds.)
  Algorithms - {ESA} 2008, 16th Annual European Symposium, Karlsruhe, Germany,
  September 15-17, 2008. Proceedings. LNCS, vol. 5193, pp. 136--147. Springer
  (2008), \url{http://dx.doi.org/10.1007/978-3-540-87744-8_12}

\bibitem{bbk-clpte-07}
Bachmaier, C., Brunner, W., K{\"{o}}nig, C.: Cyclic level planarity testing and
  embedding. In: Hong, S., Nishizeki, T., Quan, W. (eds.) Graph Drawing, 15th
  International Symposium, {GD} 2007, Sydney, Australia, September 24-26, 2007.
  Revised Papers. LNCS, vol. 4875, pp. 50--61. Springer (2008),
  \url{http://dx.doi.org/10.1007/978-3-540-77537-9_8}

\bibitem{bkr-sepg-13}
Bl\"{a}sius, T., Kobourov, S.G., Rutter, I.: Simultaneous embeddings of planar
  graphs. In: Tamassia, R. (ed.) Handbook of Graph Drawing and Visualization,
  chap.~11, pp. 349--382. Discrete Mathematics and Its Applications, Chapman
  and Hall/CRC (2013)

\bibitem{br-spacep-13}
Bl{\"{a}}sius, T., Rutter, I.: Simultaneous {PQ}-ordering with applications to
  constrained embedding problems. {ACM} Trans. Algorithms  12(2), ~16 (2016),
  \url{http://doi.acm.org/10.1145/2738054}

\bibitem{bl-tcop-76}
Booth, K.S., Lueker, G.S.: Testing for the consecutive ones property, interval
  graphs, and graph planarity using pq-tree algorithms. J. Comput. Syst. Sci.
  13(3),  335--379 (1976),
  \url{http://dx.doi.org/10.1016/S0022-0000(76)80045-1}

\bibitem{b-updsrc-14}
Brandenburg, F.: Upward planar drawings on the standing and the rolling
  cylinders. Comput. Geom.  47(1),  25--41 (2014),
  \url{http://dx.doi.org/10.1016/j.comgeo.2013.08.003}

\bibitem{bcdeeiklm-spge-07}
Bra{\ss}, P., Cenek, E., Duncan, C.A., Efrat, A., Erten, C., Ismailescu, D.,
  Kobourov, S.G., Lubiw, A., Mitchell, J.S.B.: On simultaneous planar graph
  embeddings. Comput. Geom.  36(2),  117--130 (2007),
  \url{http://dx.doi.org/10.1016/j.comgeo.2006.05.006}

\bibitem{dn-hpt-88}
{Di Battista}, G., Nardelli, E.: Hierarchies and planarity theory. {IEEE}
  Trans. Systems, Man, and Cybernetics  18(6),  1035--1046 (1988),
  \url{http://dx.doi.org/10.1109/21.23105}

\bibitem{efk-clptmp-09}
Estrella{-}Balderrama, A., Fowler, J.J., Kobourov, S.G.: On the
  characterization of level planar trees by minimal patterns. In: Eppstein, D.,
  Gansner, E.R. (eds.) Graph Drawing, 17th International Symposium, {GD} 2009,
  Chicago, IL, USA, September 22-25, 2009. Revised Papers. LNCS, vol. 5849, pp.
  69--80. Springer (2010), \url{http://dx.doi.org/10.1007/978-3-642-11805-0_9}

\bibitem{fb-clp-04}
Forster, M., Bachmaier, C.: Clustered level planarity. In: van Emde~Boas, P.,
  Pokorn{\'{y}}, J., Bielikov{\'{a}}, M., Stuller, J. (eds.) {SOFSEM} 2004:
  Theory and Practice of Computer Science, 30th Conference on Current Trends in
  Theory and Practice of Computer Science, Merin, Czech Republic, January
  24-30, 2004. LNCS, vol. 2932, pp. 218--228. Springer (2004),
  \url{http://dx.doi.org/10.1007/978-3-540-24618-3_18}

\bibitem{Fulek2013}
Fulek, R., Pelsmajer, M.J., Schaefer, M., {\v{S}}tefankovi{\v{c}}, D.:
  Hanani--{T}utte, monotone drawings, and level-planarity. In: Pach, J. (ed.)
  Thirty Essays on Geometric Graph Theory, pp. 263--287. Springer New York, New
  York, NY (2013), \url{http://dx.doi.org/10.1007/978-1-4614-0110-0_14}

\bibitem{hrl-tsp-13}
Haeupler, B., {Raju Jampani}, K., Lubiw, A.: Testing simultaneous planarity
  when the common graph is 2-connected. J. Graph Algorithms Appl.  17(3),
  147--171 (2013), \url{http://dx.doi.org/10.7155/jgaa.00289}

\bibitem{hkl-clpg-04}
Healy, P., Kuusik, A., Leipert, S.: A characterization of level planar graphs.
  Discrete Mathematics  280(1-3),  51--63 (2004),
  \url{http://dx.doi.org/10.1016/j.disc.2003.02.001}

\bibitem{hp-rlpdlt-95}
Heath, L.S., Pemmaraju, S.V.: Recognizing leveled-planar dags in linear time.
  In: Brandenburg, F. (ed.) Graph Drawing, Symposium on Graph Drawing, {GD}
  '95, Passau, Germany, September 20-22, 1995, Proceedings. Lecture Notes in
  Computer Science, vol. 1027, pp. 300--311. Springer (1995),
  \url{http://dx.doi.org/10.1007/BFb0021813}

\bibitem{hm-ppp-04}
Hsu, W., McConnell, R.M.: {PQ} trees, {PC} trees, and planar graphs. In: Mehta,
  D.P., Sahni, S. (eds.) Handbook of Data Structures and Applications. Chapman
  and Hall/CRC (2004), \url{http://dx.doi.org/10.1201/9781420035179}

\bibitem{jlm-lptlt-98}
J{\"{u}}nger, M., Leipert, S., Mutzel, P.: Level planarity testing in linear
  time. In: Whitesides, S. (ed.) Graph Drawing, 6th International Symposium,
  GD'98, Montr{\'{e}}al, Canada, August 1998, Proceedings. LNCS, vol. 1547, pp.
  224--237. Springer (1998), \url{http://dx.doi.org/10.1007/3-540-37623-2_17}

\bibitem{o-top-79}
Opatrny, J.: Total ordering problem. {SIAM} J. Comput.  8(1),  111--114 (1979)

\bibitem{rsbhkmsc-sfplg-01}
Randerath, B., Speckenmeyer, E., Boros, E., Hammer, P.L., Kogan, A., Makino,
  K., Simeone, B., Cepek, O.: A satisfiability formulation of problems on level
  graphs. Electronic Notes in Discrete Mathematics  9,  269--277 (2001),
  \url{http://dx.doi.org/10.1016/S1571-0653(04)00327-0}

\bibitem{stt-mvuhs-81}
Sugiyama, K., Tagawa, S., Toda, M.: Methods for visual understanding of
  hierarchical system structures. {IEEE} Trans. Systems, Man, and Cybernetics
  11(2),  109--125 (1981), \url{http://dx.doi.org/10.1109/TSMC.1981.4308636}

\bibitem{wsp-gktlg-12}
Wotzlaw, A., Speckenmeyer, E., Porschen, S.: Generalized k-ary tanglegrams on
  level graphs: {A} satisfiability-based approach and its evaluation. Discrete
  Applied Mathematics  160(16-17),  2349--2363 (2012),
  \url{http://dx.doi.org/10.1016/j.dam.2012.05.028}

\end{thebibliography}

\end{document}